\documentclass[11pt,a4paper,oneside,english]{article}
\usepackage{babel}
\usepackage[utf8]{inputenc}
\usepackage[T1]{fontenc} 
\usepackage{amsmath}
\usepackage{amsthm}
\usepackage{amssymb}
\usepackage{physics}
\usepackage{graphicx}
\usepackage{booktabs}
\usepackage[hidelinks]{hyperref}
\usepackage{makecell}
\usepackage{romannum}
\usepackage{longtable}
\usepackage[flushleft]{threeparttable}
\usepackage{fullpage,authblk}
\usepackage{float}
\usepackage{blindtext}
\usepackage{lmodern}
\usepackage[numbers]{natbib}

\newtheorem{theorem}{Theorem}

\providecommand{\keywords}[1]{
  \small	
  \textbf{\textit{Keywords---}} #1
}

\title{\textbf{A Novel Efficient Multiparty Semi-Quantum Secret Sharing Protocol Using Entangled States for Sharing Specific Bits}}

\begin{document}

\author[1]{Mustapha Anis Younes \footnote{Corresponding Author: \texttt{mustaphaanis.younes@univ-bejaia.dz}}}
\author[2]{Sofia Zebboudj \footnote{\texttt{sofia.zebboudj@univ-ubs.fr}}}
\author[3]{Abdelhakim Gharbi \footnote{\texttt{abdelhakim.gharbi@univ-bejaia.dz}}}
\affil[1,2]{Université de Bejaia, Faculté des Sciences Exactes, Laboratoire de Physique Théorique, 06000 Bejaia, Algérie}
\affil[2]{ENSIBS, Université de Bretagne Sud, 56000 Vannes, France}
\date{\vspace{-8ex}}

\maketitle 
\pagenumbering{arabic}
\pagenumbering{arabic}

\begin{abstract}
    Recently, Younes et al. \cite{Younes2024} proposed an efficient multi-party semi-quantum secret sharing (SQSS) scheme that generalizes Tian et al.'s three-party protocol \cite{Tian2021} to accommodate multiple participants. This scheme retains the original advantages, such as high qubit efficiency and allowing the secret dealer, Alice, to control the message content. However, He et al. \cite{He2024} identified a vulnerability in Tian et al.'s protocol to the double CNOT attack (DCNA), which also affects the generalized scheme. In response, He et al. proposed an improved protocol to address this issue. Despite these improvements, their protocol is limited to two participants and remains a primarily two-way communication scheme, which does not fully prevent the Trojan horse attack without expensive quantum devices such as photon number splitters (PNS) and wavelength filters (WF). To address these issues, this paper develops a novel multi-party SQSS scheme using the quantum property between Bell states and the Hadamard operation to detect eavesdroppers. This new scheme is secure against the DCNA, intercept-resend attack, and collective attack. It employs a fully one-way communication scheme, entirely preventing the Trojan horse attack without costly quantum devices, aligning with the semi-quantum environment's original intent. This new protocol also offers better qubit efficiency and allows Alice to share specific secrets. 
\end{abstract}

\section{Introduction}\label{intro}

Secret sharing \cite{Shamir1979} is a procedure that allows a dealer to distribute a secret among several participants. This is done in such a manner that no individual participant has any intelligible knowledge of the secret. Only when a sufficient number of participants collaborate can the dealer's secret be reconstructed \cite{Brickell1991, Brickell}. However, classical secret sharing (CSS) cannot effectively address the problem of eavesdropping unless combined with other technologies such as encryption. Moreover, the advent of quantum algorithms poses a significant threat to classical encryption algorithms \cite{Chuang1995, Plenio1996}. Fortunately, quantum cryptography can offer a solution to these issues. Unlike classical methods, the secret is split, transmitted, and recovered through quantum operations, ensuring security based on the fundamental principles of quantum physics.

In 1999, Hillery et al. \cite{Hillery1999} introduced the first quantum secret sharing (QSS) protocol using three-particle entangled GHZ states. Since then, numerous QSS protocols have been developed \cite{Zhang2005, Tyc2002, Bagherinezhad2003, Sarvepalli2012, Karlsson1999, Yu2008, Tavakoli2015, Xiang2019, Williams2019, Schmid2005, Yan2005, Qin2018}, leveraging various quantum resources such as entangled states, single photons, graph states, and multi-level states. However, these protocols typically assume that all participants in the quantum communication process possess full quantum capabilities. This assumption is highly unrealistic due to the high costs and complexity associated with the necessary quantum devices \cite{Cirac2020}.

To address these practicality challenges, Boyer et al. \cite{Boyer2007} introduced the concept of "semi-quantum environment" in 2007. This approach enables secure communication between a fully quantum-capable party and a classically limited party. The quantum participant has full quantum capabilities and can perform all quantum operations while a classical participant can only perform the following operations: (1) generate qubits in the $Z$ basis $\{\ket{0}, \ket{1}\}$; (2) measure qubits in the $Z$ basis $\{\ket{0}, \ket{1}\}$; (3) reflect qubits without disturbance; and (4) reorder qubits. Since then, many other semi-quantum protocols have emerged, such as SQKD \cite{Boyer2009,Zou2009,Boyer2017,Tsai2018}, semi-quantum private comparisons (SQPC) \cite{Jiang2020,Chongqiang2021,YanFeng2018}, and semi-quantum secret sharing protocols (SQSS) \cite{Xiang2019,Li2010,Xie2015}, In 2013, Nie et al. \cite{Nie2012} expanded this framework by introducing unitary operations on single qubits as a viable operation in a semi-quantum environment. This expansion was driven by the rapid development and feasibility of unitary operations for single qubits \cite{Bouwmeester1997,Ren2017,Young2013,Tipsmark2011,Cerf1998,O’Brien2007}. Since then, several semi-quantum protocols \cite{Tsai2020,Tsai2021,Li2020,Tsai2022,Tsai2019,Yang2020} have incorporated unitary operations alongside Boyer et al.'s original operations \cite{Iqbal2020}.

The first SQSS protocol was proposed by Li et al. \cite{Li2010} in 2010 using maximally entangled states. Two years later, Li et al. \cite{Li2013} extended this protocol to accommodate multiple participants and introduced a new SQSS protocol without entanglement, offering a more practical and feasible approach. In the same year, Yang and Hwang \cite{CHUNWEI2013} introduced a novel key construction method that significantly improved qubit efficiency. In 2015, Xie et al. \cite{Xie2015} proposed a scheme that allowed Alice to share a specific message instead of a random one. However, Yin et al. \cite{Yin2016} demonstrated that this protocol was vulnerable to intercept and resend attacks. Although they proposed a solution, it deviated from the semi-quantum condition. In the next year, Gao et al. \cite{Gao2017} provided a new fix for Xie et al.'s scheme that adhered to the semi-quantum condition, though it required each participant to share two secret keys with Alice using a secure semi-quantum QKD protocol. In 2016, a new multiparty protocol was introduced by Gao et al. \cite{Gao2016}, which restricted classical participants to performing measurements in a computational basis and rearranging the order of qubits. In 2017, Yu et al. \cite{Yu2017} introduced a novel SQSS protocol that enabled Alice to share secrets with multiple participants using $n$-particle GHZ states. This protocol could also be efficiently converted into a multi-party semi-quantum key distribution scheme. In 2020, Li et al. \cite{Li2020} proposed another multi-party SQSS scheme based on Bell states. The following year, Ye et al. \cite{Chongqiang2021} introduced a new efficient multi-party SQSS scheme based on $n$-particle GHZ states, using both measured and reflected qubits as the secret keys. This approach led to better efficiency compared to previous multi-party schemes. In 2023, Tian et al. \cite{Tian2023} introduced a scheme with a high channel capacity. During the same year, Hou et al. \cite{Hou2024} proposed a circular scheme based on hybrid single qubit and GHZ-type states. Xing et al. \cite{Xing2023} also introduced an SQSS protocol based on single photons, allowing Alice to share a specific secret by relying on particle reordering, significantly improving efficiency and ensuring security.

Recently, Tian et al. \cite{Tian2021} proposed an efficient three-party SQSS protocol based on Bell states. This protocol had significant advantages over similar schemes: (1) Alice could share specific messages with the participants, (2) it demonstrated better qubit efficiency compared to similar approaches, and (3) reduced computational overhead. However, in 2024, He et al. \cite{He2024} revealed that Tian et al.'s protocol was vulnerable to the double CNOT attack (DCNA). They showed that an eavesdropper could steal Alice's secret using this attack while remaining undetected. To address this, they proposed an improvement where classical participants reorder the qubits they return to Alice. This modification not only prevented the DCNA but also allowed Alice and the participants to share an additional random secret. Despite this, the improved protocol remains primarily a two-way scheme, which does not fully prevent the Trojan horse attack without equipping the classical participants with expensive quantum devices such as photon number splitters (PNS) and wavelength filters (WF). These devices deviate from the original intent of the semi-quantum environment, designed to demonstrate low-cost and limited-capability protocols comparable in security to fully quantum ones.

In 2023, Younes et al. \cite{Younes2024} generalized Tian et al.'s protocol to accommodate $M$ classical participants using $M$-particle entangled states instead of Bell states. This scheme partially implemented a one-way communication structure that could mitigate the Trojan horse attack to some extent. However, although the generalized scheme retains the favorable features of the original protocol, it still suffers from the same security weakness against the double CNOT attack, as it relies on the same decoy photon technology to detect eavesdroppers.

This paper addresses two main issues: (1) improving Younes et al.'s multi-party scheme to withstand the DCNA, and (2) developing a multi-party scheme that fully mitigates the Trojan horse attack without requiring costly quantum devices. To address these issues, we propose a novel multi-party SQSS protocol. In this new scheme, Alice uses $M$-particle entangled states to share her secret and Bell states to detect eavesdroppers and does not need entangled measurement capabilities. Classical participants are restricted to two operations: (a) measuring qubits in the $Z$ basis $\{\ket{0}, \ket{1}\}$ and (b) performing the Hadamard operation $H$. his improved scheme has several advantages: (i) by leveraging the property between Bell states and the Hadamard operation, the DCNA is prevented, (ii) the fully one-way communication structure makes the protocol robust against the Trojan horse attack without requiring classical participants to use costly quantum devices, (iii) Alice can control the content of the secret she shares, and (iv) the protocol offers superior qubit efficiency compared to other multi-party SQSS protocols.

This remaining of this paper is structured as follows: In Section \ref{scheme}, we introduce our new multi-party SQSS protocol. In Section \ref{security}, we present its security analysis against external and internal threats. In Section \ref{eff}, we analyze the efficiency of the scheme and offer a comparative analysis with similar multi-party SQSS protocols. Finally, the paper concludes in the \hyperref[conclusion]{last section}.

\section{The improved scheme} \label{scheme}

To address the issues mentioned in section \ref{intro}, we propose in this paper, a new semi-quantum QSS scheme, where \textit{a secret dealer}, namely Alice is an entity capable of performing quantum operations, while $M$ other participants of the scheme are capable of performing the following operations : 

\begin{enumerate}
    \item Measuring the states in the $Z$ basis.
    \begin{equation}
        Z = \{\ket{0}, \ket{1}\}
    \end{equation}

    \item The Hadamard operation represented by Eq. (\ref{eq:hadamard})
        
    \begin{align}\label{eq:hadamard}   
        H & = \frac{1}{\sqrt{2}} \Big(\ket{0}+\ket{1}\Big)\Big(\bra{0}+\bra{1}\Big)\\
        & = \frac{1}{\sqrt{2}}\Big(\ketbra{0}+\ketbra{1}{0}+\ketbra{0}{1}-\ketbra{1}{1}\Big) \nonumber \\
         & =\frac{1}{\sqrt{2}}\begin{pmatrix}
            1 & 1 \\
            1 & -1
        \end{pmatrix} \nonumber
    \end{align}
\end{enumerate}

Note that the Hadamard operations has the following actions on the states $\{\ket{0}, \ket{1}, \ket{+}, \ket{-}\}$:

\begin{align} \label{h on single qb1}
    H \ket{0} &=\ket{+} \\
    H\ket{1} &=\ket{-}\\
    H\ket{+} &=\ket{0}\\
    H\ket{-} &=\ket{1} \label{h on single qb2}
\end{align}

These two capabilities have been practiced in quantum computers and optical experiment implementations \cite{Bouwmeester1997,Ren2017,Young2013,Tipsmark2011,Cerf1998,O’Brien2007}, demonstrating the feasibility of Hadamard operations and $Z$ basis measurement devices in real-world applications.

The proposed scheme is based on our previous work \cite{Younes2024}. The novelty of this new improved scheme resides mainly in two things:
\begin{enumerate}
    \item Adapting the restrictions made on the classical participants' capabilities. 
    \item Using the propriety between Bell states and the Hadamard operation to detect eavesdroppers instead of single photons. The Bell states are defined as follow:

\begin{align} \label{bell1}
    \ket{\phi^+} &= \frac{1}{\sqrt{2}} (\ket{00}+\ket{11})=\frac{1}{\sqrt{2}}(\ket{++}+\ket{--}), \\ \label{bell2}
        \ket{\phi^-} &= \frac{1}{\sqrt{2}} (\ket{00}-\ket{11})=\frac{1}{\sqrt{2}}(\ket{+-}+\ket{-+}), \\ \label{bell3}
        \ket{\psi^+} &= \frac{1}{\sqrt{2}} (\ket{01}+\ket{10})=\frac{1}{\sqrt{2}}(\ket{++}-\ket{--}), \\ \label{bell4}
        \ket{\psi^-} &= \frac{1}{\sqrt{2}} (\ket{01}-\ket{10})=\frac{1}{\sqrt{2}}(\ket{-+}-\ket{+-}). 
\end{align}
    According to Eqs. (\ref{h on single qb1}-\ref{h on single qb2}), the Hadamard operation affects the Bell states as follows:

    \begin{align} \label{bell-hadamard1}
        (H \otimes H)\ket{\phi^+} &= \ket{\phi^+}, \\ \label{bell-hadamard2}
        (H \otimes H)\ket{\phi^-} &= \ket{\psi^-}, \\ \label{bell-hadamard3}
        (H \otimes H)\ket{\psi^+} &= \ket{\phi^-}, \\ 
        (H \otimes H)\ket{\psi^-} &=\ket{\psi^+}. \label{bell-hadamard4}
    \end{align}
\end{enumerate}
We observe from the above equations that the entanglement property of Bell states is preserved after applying the Hadamard operation to each qubit of the Bell pair. This relationship between Bell states and the Hadamard operation is utilized to detect eavesdroppers, enhancing the security of the scheme.

\subsection{Steps of our improved scheme}

Let $S$ be the classical bit string representing Alice's secret, with $N$ its length. Let also $M$ be the number of participants with whom Alice shares the secret $S$.

\subsubsection*{Step 01 : Encoding step} 

Alice first transforms the secret $S$ into the classical sequence $K_A$ by embedding $S$ with $N$ additional random bits, placed at random positions. Note that Alice saves the index of each added bit. $K_A$ is thus of length $2N$.

Based on $K_A$, Alice prepares a sequence of $2N$ entangled states of $M$ particles in one of the following states:

\begin{align}
    \ket{\psi_0} &=\frac{1}{\sqrt{2}} \Big(\ket{+}^{\otimes M} + \ket{-}^{\otimes M}\Big)\\
    & = \frac{1}{\sqrt{2}} \Big( \ket{+}_1 \otimes \ket{+}_2 \otimes ... \otimes \ket{+}_{M}  \Big) \nonumber\\
    & \hspace{1cm}  + \frac{1}{\sqrt{2}} \Big( \ket{-}_1 \otimes \ket{-}_2 \otimes ... \otimes \ket{-}_{M}  \Big) \nonumber
\end{align}

\begin{align}
    \ket{\psi_1} &=\frac{1}{\sqrt{2}} \Big(\ket{+}^{\otimes M} - \ket{-}^{\otimes M}\Big)\\
    & = \frac{1}{\sqrt{2}} \Big( \ket{+}_1 \otimes \ket{+}_2 \otimes ... \otimes \ket{+}_{M}  \Big) \nonumber\\
    & \hspace{1cm}  - \frac{1}{\sqrt{2}} \Big( \ket{-}_1 \otimes \ket{-}_2 \otimes ... \otimes \ket{-}_{M}  \Big), \nonumber
\end{align}

where, $\ket{\psi_0}$ encodes the classical bit 0 and $\ket{\psi_1}$ encodes the classical bit 1.

% The reason for this is that when all qubits in these states are measured in the $Z$ basis $\{\ket{0}, \ket{1}\}$, the measurement results $r_i$ exhibit the following correlations:
%      \begin{align}
%           0 &= r_1 \oplus r_2 \oplus \dots \oplus r_n,  \\
%           1 &= r_1 \oplus r_2 \oplus \dots \oplus r_n. \nonumber
%      \end{align}

% Depending on her preparation, Alice obtain a secret bit string $K_A$.

% She then splits her sequence into $M$ sequences $S_i = (p_i^1, p_i^2, \dots, p_i^{2N})$, where $i \in \{1, \cdots M\}$ and $p_i^j$ is the $i-th$ particle of the $j-th$ entangled state. 

\subsubsection*{Step 02 : Decoy insertion step}

Alice splits all the generated qubits into $M$ sequences in the state 

\begin{align}
    P_i = p_i^1 \otimes p_i^2 \ldots \otimes p_i^{2N},
\end{align}
where $i \in \{1, \cdots, M\}$. 

$p_i^j$ represents the state of the $i-th$ particle of the $j-th$ entangled state and $P_i$ represents the state of the qubits destined to be sent to the $i-th$ participant in step 03. 

Alice generates $M$ groups of decoy Bell pairs $G_{1}, G_{2}, \cdots, G_{M}$, where each group contains $K$ Bell pairs, each randomly chosen from the set $\{\ket{\phi^+}, \ket{\phi^-}, \ket{\psi^+}, \ket{\psi^-}\}$. For each group $G_i$ Alice divides them into two sequences of particles $D_{A_i}$ and $D_i$ such as:

\begin{align}
    D_{A_i} &= d_{A_i}^{1} \otimes d_{A_i}^{2} \ldots \otimes d_{A_i}^{K}, \\
    D_i &= d_i^1 \otimes d_i^2 \ldots \otimes d_i^{K},
\end{align}
where $d_{A_i}^{k}$ represents the state of the first particle of the $k-th$ Bell state in the $i-th$ group $G_i$, while  $d_i^k$ represents the state of the second particle of the $k-th$ Bell state in the $i-th$ group $G_i$. Alice keeps all the $M$ sequences of first particles $D_{A_i}^i$ to herself and takes each of the other $M$ sequences of second particles $D_i$, embedding them into the sequences $P_i$ at random position. This results in $M$ new sequences, denoted as $P^\prime_i$. Note that Alice keeps track of the indices of these inserted particles.

\subsubsection*{Step 03 : Eavesdropping check} 

Alice sends each $P_i^\prime$ sequence to the $i-th$ participant. After all participants confirm receipt, Alice reveals the positions of the decoy qubits in each sequence $P^\prime_i$ but not the states in which each Bell pair is prepared.

Each participant $i$ randomly performs one of the two operations on each decoy qubits:
\begin{itemize}
    \item The participant perform the Hadamard operation followed by a measurement in the $Z$ basis. We denote this as the \textbf{MH} operation.

    \item The participant directly performs a measurement in the $Z$ basis. We call this operation the \textbf{M} operation

\end{itemize}

All participants announce their operations to Alice via the public channel. Alice performs the same operations on the particles she retained. She then instructs the participants to reveal their measurement outcomes via the public channel. If no eavesdropper is present, Alice expects her measurement outcomes and the participants' measurement results to fulfill Eqs (\ref{bell1}-\ref{bell-hadamard4}).
 
 For instance, assume the Bell state between Alice and the $i$-th participant, Bob, was in the state $\ket{\phi^-}$. If both Alice and Bob perform the \textbf{M} operation on their respective qubits, Alice expects her measurement result to be correlated with Bob's measurement outcome, as shown in Eq (\ref{bell2}). Conversely, if they both perform the \textbf{MH} operation, Alice expects both measurement outcomes to be anti-correlated, as displayed in Eq (\ref{bell-hadamard2}).
 
 Alice analyzes the error rate; if it exceeds a predetermined threshold, the protocol aborts. Otherwise, the protocol proceeds to the next step.

 \subsubsection*{Step 04 : Validity check} 
 
 All participants measure the qubits of their sequences in the $Z$ basis, the measurement results $r_i$ of each state $\ket{\psi}$ exhibit the following correlations :
     \begin{align} \label{measurement correlation}
    r_1 \oplus r_2 \oplus \dots \oplus r_M = \begin{cases}
        0, & \text{if} \quad \ket{\psi} = \ket{\psi_0} \\
        1, & \text{if} \quad \ket{\psi} = \ket{\psi_1} 
    \end{cases}
      \end{align}
       resulting in each obtaining a random key sequence $K_i$. To validate the shared secret, Alice reveals the position of random $N$ bits that she placed beforehand in her sequence $K_A$. The participants then select those test bits from their respective $K_i$ and perform the XOR operation to verify the relationship :

 \begin{equation}
     K_{A}^{test} = K_{1}^{test} \oplus K_{2}^{test} \oplus \dots \oplus K_{n}^{test}.
 \end{equation}
 
 They declare their results to Alice via the public channel. If the relationship holds, they discard the test bits and use the remaining ones to recover Alice's secret $S$ using the relation: 

 \begin{equation} \label{secret recovery}
     K_{A} = K_{1} \oplus K_{2} \oplus \dots \oplus K_{n}.
 \end{equation}
Note that the participants must not use the public channel to reveal their individual shares; otherwise, an eavesdropper would be able to recover Alice's secret. The recovery of Alice's secret $S$ S could be done, for example, by requiring the participants to meet in person to exchange their shares.

\section{Security analysis} \label{security}

In this section, we analyze the security of the proposed protocol against both external and internal attacks. The attacker's goal in both cases is to steal Alice's secret $S$ without being detected. Since our protocol is an ($M$, $M$) threshold scheme, the attacker must obtain all participants' shares to recover the secret.  We demonstrate that the protocol can withstand external or internal threat.

\subsection{External attack}
In this scenario, we consider an external eavesdropper, Eve, who attempts to obtain Alice's secret $S$ without detection. To do so, Eve would need to know all participants' shares or the initial states prepared by Alice. To achieve this, She might employ well-known attacks such as the double CNOT attack, intercept-and-resend attack, collective attack, or Trojan horse attack. Since this protocol adopts a one-way communication via the quantum channel, Eve can only launch her attack on the transmitted particles in step 3 when the Alice sends the sequences $P_{i}^\prime$ to the participants. Given that this protocol relies solely on decoy Bell states to detect eavesdroppers, we focus on the impact of these attacks on the decoy states.

\subsubsection{Double CNOT attack}

The double CNOT attack (DCNA) involves using two CNOT gates to extract information from a quantum system. When the sender transmits a qubit to the receiver, Eve intercepts it and performs a CNOT gate, using the transmitted qubit as the control and her ancillary qubit (initially in the state $\ket{0}$) as the target. We define the CNOT gate as follow:

\begin{align}
    U_{CNOT} &= \ketbra{00} + \ketbra{01}{01} + \ketbra{11}{10} + \ketbra{10}{11}, \\
             &= \begin{pmatrix}
            1 & 0 & 0 & 0 \\
            0 & 1 & 0 & 0 \\
            0 & 0 & 0 & 1 \\
            0 & 0 & 1 & 0
        \end{pmatrix}. \nonumber
\end{align}
This operation entangles the transmitted qubit with the ancilla.  After receiving the qubit, the receiver may perform some operation on the received qubit and return it to the sender. Eve intercepts it again and applies a second CNOT gate, using the intercepted qubit as the control and the ancilla as the target. This disentangles the qubit from the ancilla. Eve then measures her ancillary qubit in the appropriate basis to gain information about the state of the transmitted qubit.

In the proposed protocol, participants do not return any particles to Alice via the quantum channel, limiting Eve to performing the CNOT gate only once. Without loss of generality, consider a scenario where Alice shares her secret with three participants: Bob, Charlie, and Diana. She encodes her classical bits "0" and "1" respectively in the following states:

\begin{align}
    \ket{\psi_0} &= \frac{1}{\sqrt{2}} \big(\ket{+++} + \ket{---}\big), \nonumber \\
                 &= \frac{1}{\sqrt{2}} \big(\ket{000} + \ket{011} + \ket{101} + \ket{110}\big).
\end{align}

\begin{align}
     \ket{\psi_1} &= \frac{1}{\sqrt{2}} \big(\ket{+++} - \ket{---}\big), \nonumber \\
                  &= \frac{1}{\sqrt{2}} \big(\ket{001} + \ket{010} + \ket{100} + \ket{111}\big).
\end{align}
Eve intercepts each qubit sent from Alice to the participants via the quantum channel. She generates an ancillary qubit in the state $\ket{0}_e$ for each intercepted particle and performs the CNOT gate,  where Alice's qubit is the control qubit and Eve's qubit $\ket{0}_e$ is the target qubit. According to Alice's quantum state, the qubit systems become the following:

\begin{align} \label{cnot on phi0}
    U_{CNOT} \ket{\psi_0} &= \frac{1}{\sqrt{2}} \big(\ket{000}_{bcd}\ket{000}_e + \ket{011}_{bcd}\ket{011}_e + \ket{101}_{bcd}\ket{101}_e + \ket{110}_{bcd}\ket{110}_e\big), \\ \label{cnot on phi1}
    U_{CNOT} \ket{\psi_1} &= \frac{1}{\sqrt{2}} \big(\ket{001}_{bcd}\ket{001}_e + \ket{010}_{bcd}\ket{010}_e + \ket{100}_{bcd}\ket{100}_e + \ket{111}_{bcd}\ket{111}_e\big), \\ \label{cnot on phi+}
    U_{CNOT} \ket{\phi^+} &= \frac{1}{\sqrt{2}} \big(\ket{00}_{ap}\ket{0}_e + \ket{11}_{ap}\ket{1}_e \big), \\ \label{cnot on phi-}
    U_{CNOT} \ket{\phi^-} &= \frac{1}{\sqrt{2}} \big(\ket{00}_{ap}\ket{0}_e - \ket{11}_{ap}\ket{1}_e \big), \\ \label{cnot on psi+}
    U_{CNOT} \ket{\psi^+} &= \frac{1}{\sqrt{2}} \big(\ket{01}_{ap}\ket{1}_e + \ket{10}_{ap}\ket{0}_e \big), \\ \label{cnot on psi-}
    U_{CNOT} \ket{\psi^-} &= \frac{1}{\sqrt{2}} \big(\ket{01}_{ap}\ket{1}_e + \ket{10}_{ap}\ket{0}_e \big), 
\end{align}
where the subscript $b$, $c$, $d$, $e$ refer to the qubits of Bob, Charlie, Diana, and Eve, respectively, and the subscript $ap$ indicates the decoy Bell state that Alice can share with one of the participants.

In the following scenario, we observe from Eqs. (\ref{cnot on phi0}) and (\ref{cnot on phi1}) that when the participants measure their respective qubits in the $Z$ basis, they obtain one of the following outcomes for each state:

\begin{align}
    U_{CNOT} \ket{\psi_0} \stackrel{Measure}{\longrightarrow} \begin{cases}
        \ket{000}_{bcd}\ket{000}_{e},  \\
        \ket{011}_{bcd}\ket{011}_{e},  \\
        \ket{101}_{bcd}\ket{101}_{e},  \\
        \ket{110}_{bcd}\ket{110}_{e}. 
    \end{cases}
      \end{align}

\begin{align}
    U_{CNOT} \ket{\psi_1} \stackrel{Measure}{\longrightarrow} \begin{cases}
        \ket{001}_{bcd}\ket{001}_{e},  \\
        \ket{010}_{bcd}\ket{010}_{e},  \\
        \ket{100}_{bcd}\ket{100}_{e},  \\
        \ket{111}_{bcd}\ket{111}_{e}.  
    \end{cases}
      \end{align}
      From the two above equations, we notice two things:
      \begin{enumerate}
          \item The measurement results of the participants exhibit the correlations shown in Eq. (\ref{measurement correlation}).
          \item The state of Eve's ancillary qubits is always correlated with the participants' measurement results, meaning Eve can simply measure her ancillary qubits in the $Z$ basis to obtain the participants' individual shares. However, her attack will inevitably be detected with the help of the decoy pairs. 
      \end{enumerate}

In step 3 of the protocol, during the eavesdropping check, Alice reveals the positions of the decoy particles. Each participant then randomly and independently performs either the \textbf{MH} operation, where they apply the Hadamard operation followed by a $Z$ basis measurement on the decoy qubit, or the \textbf{M} operation, where they directly measure the received qubit in the $Z$ basis. Participants reveal their operations to Alice via the public channel, and Alice performs the same operation on her part of the decoy pair. This results in two situations for each participant:  

\begin{enumerate}
    \item \textbf{Situation 01:} Both Alice and the participant (let's assume it's Bob) perform the \textbf{M} operation on their respective parts of the Bell pair. In this situation, we have the following possible measurement outcomes according to Eqs. (\ref{cnot on phi+}-\ref{cnot on psi-}):
\begin{align}
    U_{CNOT} \ket{\phi^+} & \stackrel{Measure}{\longrightarrow} \begin{cases}
        \ket{00}_{ap}\ket{0}_{e},  \\
        \ket{11}_{ap}\ket{1}_{e}.
    \end{cases} \\
     U_{CNOT} \ket{\phi^-} & \stackrel{Measure}{\longrightarrow} \begin{cases}
        \ket{00}_{ap}\ket{0}_{e},  \\
        \ket{11}_{ap}\ket{1}_{e}.
    \end{cases} \\ 
     U_{CNOT} \ket{\psi^+} & \stackrel{Measure}{\longrightarrow} \begin{cases}
        \ket{01}_{ap}\ket{1}_{e},  \\
        \ket{10}_{ap}\ket{0}_{e}.
    \end{cases} \\
     U_{CNOT} \ket{\psi^-} & \stackrel{Measure}{\longrightarrow} \begin{cases}
        \ket{01}_{ap}\ket{1}_{e},  \\
        \ket{10}_{ap}\ket{0}_{e}.
    \end{cases}
      \end{align}
From the above equations, we observe that Alice's and Bob's measurement outcomes fulfill Eqs. (\ref{bell1}-\ref{bell4}) as expected. Therefore, Eve goes undetected when Alice and Bob both perform the \textbf{M} operation.

\item \textbf{Situation 02:} Both Alice and Bob perform the \textbf{MH} operation. After applying the Hadamard operation, Eqs. (\ref{cnot on phi+}-\ref{cnot on psi-}) become:

\begin{align}
     (H \otimes H \otimes I) U_{CNOT} \ket{\phi^+} &= \frac{1}{\sqrt{2}} \big(\ket{++}_{ap}\ket{0}_e + \ket{--}_{ap}\ket{1}_e \big), \nonumber \\
           &= \frac{1}{2} \Biggl(\frac{(\ket{00}+ \ket{11})_{ap}}{\sqrt{2}}\ket{+}_{e} + \frac{(\ket{01}+ \ket{10})_{ap}}{\sqrt{2}}\ket{-}_{e}\Biggl).
\end{align}
 \begin{align}
      (H \otimes H \otimes I) U_{CNOT} \ket{\phi^-} &= \frac{1}{\sqrt{2}} \big(\ket{++}_{ap}\ket{0}_e - \ket{--}_{ap}\ket{1}_e \big), \nonumber \\
           &= \frac{1}{2} \Biggl(\frac{(\ket{00}+ \ket{11})_{ap}}{\sqrt{2}}\ket{-}_{e} + \frac{(\ket{01}+\ket{10})_{ap}}{\sqrt{2}}\ket{+}_{e}\Biggl).
 \end{align}
\begin{align}
    (H \otimes H \otimes I) U_{CNOT} \ket{\psi^+} &= \frac{1}{\sqrt{2}} \big(\ket{+-}_{ap}\ket{0}_e + \ket{-+}_{ap}\ket{1}_e \big), \nonumber \\
           &= \frac{1}{2} \Biggl(\frac{(\ket{00}- \ket{11})_{ap}}{\sqrt{2}}\ket{+}_{e} + \frac{(\ket{01}- \ket{10})_{ap}}{\sqrt{2}}\ket{-}_{e}\Biggl).
\end{align}
\begin{align}
    (H \otimes H \otimes I) U_{CNOT} \ket{\phi^+} &= \frac{1}{\sqrt{2}} \big(\ket{++}_{ap}\ket{0}_e + \ket{--}_{ap}\ket{1}_e \big), \nonumber \\
           &= -\frac{1}{2} \Biggl(\frac{(\ket{00}- \ket{11})_{ap}}{\sqrt{2}}\ket{-}_{e} + \frac{(\ket{01}- \ket{10})_{ap}}{\sqrt{2}}\ket{+}_{e}\Biggl).
\end{align}
      From these equations, we see that when Alice and Bob perform their $Z$ basis measurement, they have an equal probability of obtaining one of the four states $\{\ket{00}, \ket{01}, \ket{10}, \ket{11}\}$. Thus, for each Bell state, Alice and Bob have a $\frac{1}{2}$ probability of obtaining a result that fulfills Eqs. (\ref{bell-hadamard1}-\ref{bell-hadamard4}). Consequently, when Alice and Bob choose the \textbf{MH} operation, Eve has a $\frac{1}{2}$ probability of being detected.
\end{enumerate}

Overall, the probability of Eve being detected when launching the CNOT attack is $1-(\frac{1}{4})^K$, where $K$ is the number of decoy Bell pair for each participant. This probability approaches 1 when $K$ is sufficiently large.

\subsubsection{Intercept and resend attack}

In this attack, Eve may take two different strategies to obtain useful information about the participants' shares:

\begin{enumerate}
    \item \textbf{Eavesdropping strategy 01:} In this first strategy, Eve intercepts each sequence $P_{i}\prime$ transmitted from Alice to the participants. She stores them in her quantum memory and sends previously prepared fake particle sequences to the participants. Eve then waits for Alice to announce the positions of the decoy particles and test bits, discards these from the sequences $P_{i}^\prime$ and measures the remaining particles in the $Z$ basis to recover Alice's secret $S$ using Eq. (\ref{secret recovery}). However, this strategy fails because Eve does not know the positions and states of the decoy particles when she sends the fake sequences. By replacing the sequences $P_{i}^\prime$ with fake ones, she destroys the correlation properties between Alice's and the participants' measurement outcomes on the decoy Bell pairs. Therefore, any fake sequences Eve sends will introduce errors during the eavesdropping check, exposing her presence.

    \item \textbf{Eavesdropping strategy 02:} In this second strategy, Eve intercepts the sequences $P_{i}^\prime$, measures them in the $Z$ basis, and sends the measured sequences to the participants. She records her measurement outcomes and ends up with $M$ bit string sequences. After Alice announces the positions of the decoy particles and test bits, Eve discards those positions from her sequences and uses the remaining measurement results to reconstruct Alice's secret $S$ using Eq. (\ref{secret recovery}). However, this strategy also fails because Eve cannot distinguish the decoy particles from the entangled qubits in the sequences $P_{i}^\prime$, as they are all in the maximally mixed state:
    \begin{align}
        \rho = \frac{1}{2} \big(\ketbra{0} + \ketbra{1}\big).
    \end{align}
    Therefore, Eve will be detected with a probability of $\frac{1}{2}$ when Alice and a participant perform the \textbf{MH} operation. Overall, this eavesdropping strategy has a detection probability of $1-(\frac{1}{4})^K$, where $K$ is the number of decoy Bell state for each participant. This probability approaches 1 when $K$ is sufficiently large. 

    In fact, Eve can implement measurements in any basis she chooses, but since she cannot distinguish the decoy particles from the entangled qubits, she will inevitably introduce errors that will be detected during the eavesdropping check.

\end{enumerate}

\subsubsection{Collective attack}

We define the collective attack as follows \cite{Iqbal2020, Krawec2015}:

\begin{itemize}
    \item Eve entangles an ancillary qubit with each qubit sent from Alice and the participants through the quantum channels. Eve then measures these ancillary qubits to gather information about the participants' secret shares.
    \item Each qubit sent by Alice is attacked by Eve independently using the same strategy.
    \item Eve is free to keep her ancillary qubits in a quantum memory and measure them at an ulterior time.
\end{itemize}

As mentioned before, since this protocol uses decoy Bell states to detect eavesdroppers, we only consider the effect of the collective attack on those decoy Bell states. The security analysis is quite similar to the security analysis of the mediated semi-quantum QKD protocol of Tsai and Yang \cite{Tsai2021}.

\begin{theorem}
    Assume Eve performs the collective attack on the qubits sent from Alice to the participants by applying a unitary operation $U_E$ on the transmitted qubits and their associated ancillary qubits, such that $U^{\dagger}_{e}U_{e}= I$, where $I$ is the identity matrix. Eve cannot obtain any valuable information about the participants' secret shares without disturbing the original quantum systems, resulting in a nonzero probability of being detected during the eavesdropping check.
\end{theorem}

\begin{proof}

In her attack, Eve intercepts the sequences $P_{i}^\prime$ when they are transmitted from Alice to participants through the quantum channel. She then apply a unitary operation $U_E$ on each qubit of the sequences with its prepared ancillary qubit $E$, initially in the the state $\ket{e}$, to entangle them. The unitary operation $U_E$ is defined as follow:

\begin{align} \label{U_E on 0}
    U_{E} (\ket{0} \otimes \ket{e}) &= a \ket{0}\ket{e_{00}} + b\ket{1} \ket{e_{01}}, \\ \label{U_E on 1}
    U_{E} (\ket{1} \otimes \ket{e}) &= c \ket{0}\ket{e_{10}} + d \ket{1}\ket{e_{11}},
\end{align}
such that $\abs{a}^{2}+\abs{b}^{2}=1$, $\abs{c}^{2}+\abs{d}^{2}=1$ and the states $\ket{e_{ij}}$ are orthogonal states that Eve can distinguish.

In this proof, we consider, without loss of generalization, the case where Eve intercepts Bob's sequence $P_{B}^\prime$. Since Eve uses the same strategy to attack the other sequences $P_{i}^\prime$ this proof is valid for her attacks on the sequences of all participants. We denote the qubits kept by Alice by $A$ and those sent to Bob by $B$. The states of the decoy Bell pairs before the attack is as follow:

\begin{align}
    \ket{\phi^+}\otimes\ket{e} &= \frac{1}{\sqrt{2}} \Big(\ket{0}_A\ket{0}_B\ket{e}+\ket{1}_A\ket{1}_B\ket{e}\Big), \\
    \ket{\phi^-}\otimes\ket{e} &= \frac{1}{\sqrt{2}} \Big(\ket{0}_A\ket{0}_B\ket{e}-\ket{1}_A\ket{1}_B\ket{e}\Big), \\
    \ket{\psi^+}\otimes\ket{e} &= \frac{1}{\sqrt{2}} \Big(\ket{0}_A\ket{1}_B\ket{e}+\ket{1}_A\ket{0}_B\ket{e}\Big), \\
    \ket{\phi^-}\otimes\ket{e} &= \frac{1}{\sqrt{2}} \Big(\ket{0}_A\ket{1}_B\ket{e}-\ket{1}_A\ket{0}_B\ket{e}\Big).
\end{align}
After Eve has performed $U_E$ on the system $B+E$, the above equations evolve as follow:

\begin{align}
    (I \otimes U_E) \ket{\phi^+}\otimes\ket{e} &= \frac{1}{\sqrt{2}} \Big(\ket{0}_A\big(U_E\ket{0}_B\ket{e}\big)+\ket{1}_A\big(U_E\ket{1}_B\ket{e}\big)\Big), \nonumber \\
    &= \frac{1}{\sqrt{2}} \Big(a\ket{00}_{AB}\ket{e_{00}} + b\ket{01}_{AB}\ket{e_{01}} + c\ket{10}_{AB}\ket{e_{10}} + d\ket{11}_{AB}\ket{e_{11}} \Big).
\end{align}

\begin{align}
    (I \otimes U_E) \ket{\phi^-}\otimes\ket{e} &= \frac{1}{\sqrt{2}} \Big(\ket{0}_A\big(U_E\ket{0}_B\ket{e}\big)-\ket{1}_A\big(U_E\ket{1}_B\ket{e}\big)\Big), \nonumber \\
    &=\frac{1}{\sqrt{2}} \big(a\ket{00}_{AB}\ket{e_{00}} + b\ket{01}_{AB}\ket{e_{01}} - c\ket{10}_{AB}\ket{e_{10}} - d\ket{11}_{AB}\ket{e_{11}} \big).
\end{align}

\begin{align}
    (I \otimes U_E) \ket{\psi^+}\otimes\ket{e} &= \frac{1}{\sqrt{2}} \Big(\ket{0}_A\big(U_E\ket{1}_B\ket{e}\big)+\ket{1}_A\big(U_E\ket{0}_B\ket{e}\big)\Big), \nonumber \\
    &=\frac{1}{\sqrt{2}} \big(c\ket{00}_{AB}\ket{e_{10}} + d\ket{01}_{AB}\ket{e_{11}} + a\ket{10}_{AB}\ket{e_{00}} + b\ket{11}_{AB}\ket{e_{01}} \big).
\end{align}

\begin{align}
    (I \otimes U_E) \ket{\phi^-}\otimes\ket{e} &= \frac{1}{\sqrt{2}} \Big(\ket{0}_A\big(U_E\ket{1}_B\ket{e}\big)-\ket{1}_A\big(U_E\ket{0}_B\ket{e}\big)\Big), \nonumber \\
    &=\frac{1}{\sqrt{2}} \big(c\ket{00}_{AB}\ket{e_{10}} + d\ket{01}_{AB}\ket{e_{11}} - a\ket{10}_{AB}\ket{e_{00}} - b\ket{11}_{AB}\ket{e_{01}} \big).
\end{align}

These are the states that Alice, Bob, and Eve share among themselves. We now consider two situations:

\begin{enumerate}
    \item \textbf{Situation 01:} Alice and Bob directly measure their respective qubits in the $Z$ basis, i.e., perform the \textbf{M} operation.
    \item \textbf{Situation 02:} Alice and Bob perform the Hadamard operation $H$ on their respective qubits before performing the $Z$ basis measurement, i.e., perform the \textbf{MH} operation.
\end{enumerate}
When both Alice and Bob perform the Hadamard operation on their respective qubits, the above states evolve, respectively, as follows:

\begin{align} \label{HH on UE 1}
    \frac{1}{2\sqrt{2}}\left(\begin{array}{c}
\ket{00}_{AB} \otimes \big(a\ket{e_{00}} + b\ket{e_{01}} + c\ket{e_{10}} + d\ket{e_{11}}\big)+ \\
\ket{01}_{AB} \otimes \big(a\ket{e_{00}} - b\ket{e_{01}} + c\ket{e_{10}} - d\ket{e_{11}}\big)+ \\
\ket{10}_{AB} \otimes \big(a\ket{e_{00}} + b\ket{e_{01}} - c\ket{e_{10}} - d\ket{e_{11}}\big)+ \\
\ket{11}_{AB} \otimes \big(a\ket{e_{00}} - b\ket{e_{01}} - c\ket{e_{10}} + d\ket{e_{11}}\big)
\end{array}\right)
\end{align}

\begin{align} \label{HH on UE 2}
    \frac{1}{2\sqrt{2}}\left(\begin{array}{c}
\ket{00}_{AB} \otimes \big(a\ket{e_{00}} + b\ket{e_{01}} - c\ket{e_{10}} - d\ket{e_{11}}\big)+ \\
\ket{01}_{AB} \otimes \big(a\ket{e_{00}} - b\ket{e_{01}} - c\ket{e_{10}} + d\ket{e_{11}}\big)+ \\
\ket{10}_{AB} \otimes \big(a\ket{e_{00}} + b\ket{e_{01}} + c\ket{e_{10}} + d\ket{e_{11}}\big)+ \\
\ket{11}_{AB} \otimes \big(a\ket{e_{00}} - b\ket{e_{01}} + c\ket{e_{10}} - d\ket{e_{11}}\big)
\end{array}\right)
\end{align}

\begin{align}\label{HH on UE 3}
    \frac{1}{2\sqrt{2}}\left(\begin{array}{c}
\ket{00}_{AB} \otimes \big(a\ket{e_{00}} + b\ket{e_{01}} + c\ket{e_{10}} + d\ket{e_{11}}\big)+ \\
\ket{01}_{AB} \otimes \big(a\ket{e_{00}} - b\ket{e_{01}} + c\ket{e_{10}} - d\ket{e_{11}}\big)+ \\
\ket{10}_{AB} \otimes \big(-a\ket{e_{00}} - b\ket{e_{01}} + c\ket{e_{10}} + d\ket{e_{11}}\big)+ \\
\ket{11}_{AB} \otimes \big(-a\ket{e_{00}} + b\ket{e_{01}} + c\ket{e_{10}} - d\ket{e_{11}}\big)
\end{array}\right)
\end{align}

\begin{align} \label{HH on UE 4}
    \frac{1}{2\sqrt{2}}\left(\begin{array}{c}
\ket{00}_{AB} \otimes \big(-a\ket{e_{00}} - b\ket{e_{01}} + c\ket{e_{10}} + d\ket{e_{11}}\big)+ \\
\ket{01}_{AB} \otimes \big(-a\ket{e_{00}} - b\ket{e_{01}} + c\ket{e_{10}} - d\ket{e_{11}}\big)+ \\
\ket{10}_{AB} \otimes \big(a\ket{e_{00}} + b\ket{e_{01}} + c\ket{e_{10}} + d\ket{e_{11}}\big)+ \\
\ket{11}_{AB} \otimes \big(a\ket{e_{00}} - b\ket{e_{01}} + c\ket{e_{10}} - d\ket{e_{11}}\big)
\end{array}\right)
\end{align}

In step 3 of the proposed protocol, during the eavesdropping check, after Alice and the participants have performed their operations, Bob reveals his measurement results for each Bell state, and Alice compares them with her own. Alice expects her and Bob's measurement results to satisfy Eqs. (\ref{bell1}-\ref{bell-hadamard4}). For Eve to ensure that, she must adjust $U_E$ accordingly.

If Eve adjusts $U_E$ for the first situation, she must set:

\begin{equation} \label{1st condition}
    b = c = 0,
\end{equation}
to pass the public discussion. However, these conditions imply that the states of the quantum systems in the second scenario will be given as follows:

\begin{align}
    \frac{1}{2\sqrt{2}}\left(\begin{array}{c}
\ket{00}_{AB} \otimes \big(a\ket{e_{00}} + d\ket{e_{11}}\big)+ \\
\ket{01}_{AB} \otimes \big(a\ket{e_{00}} - d\ket{e_{11}}\big)+ \\
\ket{10}_{AB} \otimes \big(a\ket{e_{00}} - d\ket{e_{11}}\big)+ \\
\ket{11}_{AB} \otimes \big(a\ket{e_{00}} + d\ket{e_{11}}\big)
\end{array}\right)
\end{align}

\begin{align}
    \frac{1}{2\sqrt{2}}\left(\begin{array}{c}
\ket{00}_{AB} \otimes \big(a\ket{e_{00}} - d\ket{e_{11}}\big)+ \\
\ket{01}_{AB} \otimes \big(a\ket{e_{00}} + d\ket{e_{11}}\big)+ \\
\ket{10}_{AB} \otimes \big(a\ket{e_{00}} + d\ket{e_{11}}\big)+ \\
\ket{11}_{AB} \otimes \big(a\ket{e_{00}} - d\ket{e_{11}}\big)
\end{array}\right)
\end{align}

\begin{align}
    \frac{1}{2\sqrt{2}}\left(\begin{array}{c}
\ket{00}_{AB} \otimes \big(a\ket{e_{00}} + d\ket{e_{11}}\big)+ \\
\ket{01}_{AB} \otimes \big(a\ket{e_{00}} - d\ket{e_{11}}\big)+ \\
\ket{10}_{AB} \otimes \big(-a\ket{e_{00}} + d\ket{e_{11}}\big)+ \\
\ket{11}_{AB} \otimes \big(-a\ket{e_{00}} - d\ket{e_{11}}\big)
\end{array}\right)
\end{align}

\begin{align}
    \frac{1}{2\sqrt{2}}\left(\begin{array}{c}
\ket{00}_{AB} \otimes \big(-a\ket{e_{00}} + d\ket{e_{11}}\big)+ \\
\ket{01}_{AB} \otimes \big(-a\ket{e_{00}} - d\ket{e_{11}}\big)+ \\
\ket{10}_{AB} \otimes \big(a\ket{e_{00}} + d\ket{e_{11}}\big)+ \\
\ket{11}_{AB} \otimes \big(a\ket{e_{00}} - d\ket{e_{11}}\big)
\end{array}\right)
\end{align}

To pass the security check, Eve must set the incorrect terms as a zero vector. This requires meeting the condition:

\begin{equation} \label{2nd condition}
    a\ket{e_{00}} - d\ket{e_{11}} = 0,
\end{equation}
which signifies $a\ket{e_{00}} = d\ket{e_{11}}$, which contradicts the orthogonality of $\ket{e_00}$ and $\ket{e_11}$, unless Eve sets $a = d = 0$. However, combining this with Eq. (\ref{1st condition}), it conflicts with:

\begin{equation} \label{normalization}
    \abs{a}^{2}+\abs{b}^{2}=\abs{c}^{2}+\abs{d}^{2}=1.
\end{equation}
In this scenario, where Eve cannot distinguish the vectors $\ket{e_{00}}$ and $\ket{e_{11}}$ and where $b=c=0$, if we set $a=d=1$ to satisfy Eq. (\ref{normalization}), the Eqs. (\ref{U_E on 0})-(\ref{U_E on 1}) can be written as:

\begin{align} \label{new U_E on 0}
    U_{E} (\ket{0} \otimes \ket{e}) &= \ket{0}\ket{e_{00}}, \\ \label{new U_E on 1}
    U_{E} (\ket{1} \otimes \ket{e}) &= \ket{1}\ket{e_{00}}.
\end{align}
We notice that Eve's ancilla are independent of the qubits attacked. Therefore, Eve cannot measure her ancillary qubit to extract information about Bob's measurement results without having a nonzero probability of being detected.

Suppose now that Eve adjusts $U_E$ to the second situation, implying that she must eliminate the incorrect terms in Eqs. (\ref{HH on UE 1}-\ref{HH on UE 4}). Thus, she must set:

\begin{equation}
    \begin{aligned}
    a\ket{e_{00}} - b\ket{e_{01}} + c\ket{e_{10}} - d\ket{e_{11}} &=0 \\
    a\ket{e_{00}} + b\ket{e_{01}} - c\ket{e_{10}} - d\ket{e_{11}} &=0.
\end{aligned}
\end{equation}

\begin{equation}
    \begin{aligned}
        a\ket{e_{00}} + b\ket{e_{01}} - c\ket{e_{10}} - d\ket{e_{11}} &=0 \\
    a\ket{e_{00}} - b\ket{e_{01}} + c\ket{e_{10}} - d\ket{e_{11}} &=0.
    \end{aligned}
\end{equation}

\begin{equation}
    \begin{aligned}
        a\ket{e_{00}} - b\ket{e_{01}} + c\ket{e_{10}} - d\ket{e_{11}} &=0 \\
    -a\ket{e_{00}} - b\ket{e_{01}} + c\ket{e_{10}} + d\ket{e_{11}} &=0.
    \end{aligned}
\end{equation}

\begin{equation}
    \begin{aligned}
        -a\ket{e_{00}} - b\ket{e_{01}} + c\ket{e_{10}} + d\ket{e_{11}} &=0 \\
    a\ket{e_{00}} - b\ket{e_{01}} + c\ket{e_{10}} - d\ket{e_{11}} &=0.
    \end{aligned}
\end{equation}

All those equations lead to the same condition:

\begin{equation}
        a\ket{e_{00}} - d\ket{e_{11}} = 0,
\end{equation}
This means Eve must set $b=c=0$ to avoid detection in the first situation. Therefore, whether Eve adjusts $U_E$ for the first or second situation, she ends up with the same conditions that leads to Eqs. (\ref{new U_E on 0}) and (\ref{new U_E on 1}).

\end{proof}

According to the above analysis, although Eve can determine a collective attack path $U_E$ that remains undetected by Alice, she cannot obtain any information regarding the participants' secret shares because $a\ket{e_{00}} = d\ket{e_{11}}$ and $b = c = 0$. Conversely, if Eve wishes to obtain useful information about the participants' secret shares, she cannot use $U_E$ to execute a collective attack. Doing so would induce a detectable disturbance, increasing the error rate and giving Alice a nonzero probability of detecting Eve's attack.

\subsubsection{Trojan horse attack}

In this subsection, we discuss the protocol's robustness against two types of Trojan horse attacks. The first type involves Eve using a "delayed photon", where she intercepts the qubits transmitted from Alice to the participants and attaches a probing photon with a delay time shorter than the time window. The delayed photon remains undetected by the participants' detectors. After a participant performs his corresponding operation and returns the qubit, Eve intercepts it again, separates the probing photon, and gains information about the participant's operation by measuring it. In the second type of attack, Eve attaches an "invisible photon" to each intercepted qubit. These photons are undetectable by the participants. When a participant performs an operation on his qubit, the invisible photon undergoes the same operation. Eve retrieves these photons and measures them to obtain information about the participants' operations. In both scenarios, Eve can extract full information about the participants' operations when retrieving the Trojan horse photons.

In the aforementioned attack methods, the attacker can only extract information about participant operations when they retrieve the Trojan horse photons. Two-way or circular communication protocols provide the attacker an opportunity to retrieve malicious photons introduced through these attacks. Thus, such protocols are considered vulnerable to Trojan horse attacks. In contrast, a one-way communication protocol does not offer the attacker a chance to collect the malicious photons, as no qubit is returned by the participant. Consequently, a one-way communication protocol is robust against Trojan horse attacks.

In the proposed protocol, a one-way communication scheme is adopted via the quantum channel, where qubits are sent only from Alice to the classical participants. Although Eve can insert probing photons into each transmitted qubit, the participants' secret messages remain secure because Eve cannot retrieve the probing photons she inserted. Thus, the protocol is robust against Trojan horse attacks without requiring participants to equip themselves with expensive devices, such as photon number splitters and optical wavelength filters, to mitigate these attacks.

\subsection{Internal attack}

In this subsection, we examine  the scenario where the attack comes from the inside. Dishonest participants are particularly dangerous because they actively participate in the protocol. Therefore, as mentioned in \cite{Qin2007}, internal threats require closer attention in the security analysis of SQSS protocols. We consider two cases of internal threats in this study. The first case is when a single malicious participant attempts to obtain Alice's secret $S$ without the help of the other participants. The second case is the extreme scenario where $M-1$ dishonest participants collude to obtain the secret without the collaboration of the honest participant, Bob.

\subsubsection{Participant's attack}
In the proposed protocol, all participants have identical roles. Without loss of generality, we assume Charlie possesses quantum capabilities and attempts to learn the participants' secret shares by attacking their sequences $P_{i}\prime$, where $i \in \{1, 3, 4, \cdots, M\}$. Since the protocol uses one-way communication via the quantum channel, Charlie can only launch her attack in step 3 when Alice transmits the sequence $P_{i}\prime$ to the participants. In this scenario, Charlie acts exactly as an external eavesdropper since she doesn't know the position of the decoy qubits in the sequence $P_{i}\prime$, as each particle in these sequences is in the maximally mixed state $\rho = (\ket{0}\bra{0} + \ket{1}\bra{1})/2$. Therefore, Charlie might launch attacks such the DCNA attack, intercept and resend attack, collective attack, the Trojan horse, at the end her attack will inevitably introduce errors and be detected by Alice as analyzed in the previous subsection. Thus, an attack from one malicious participant is ineffective in this protocol.

\subsubsection{Collusion attack}

In this scenario, we consider the extreme case where Bob is the only honest participant in the scheme, and the other $M-1$ dishonest participants collude to obtain Alice's secret $S$ without his help. 

A first strategy would be for the $M-1$ malicious participants to deduce Alice's secret by combining their respective key sequence $K_i$, such as $i \in \{2, 3, \cdots, N\}$, they obtained at the end of the protocol. However, since the proposed protocol is an $(M, M)$ threshold scheme, they cannot infer Alice's secret without Bob's sequence $K_1$. To illustrate this, let's consider the case where the $M-1$ participants try to infer one bit of Alice's secret:

\begin{align}
    K_1^j \oplus K_2^j \oplus \dots \oplus K_M^j = \begin{cases}
        x \oplus 0 = x,  \\
        x \oplus 1 = \Bar{x}.
    \end{cases}
      \end{align}
We see from this equation that the $M-1$ dishonest participants have a probability of $\frac{1}{2}$ of guessing the correct bit. Therefore, we can deduce that they have a total probability of $\frac{1}{2}^N$ of guessing Alice's entire secret $S$. this probability converges to 0 when the length of the secret $N$ becomes large enough.

A second strategy involves launching an attack on the sequence $P_{1}\prime$ that Alice transmits to Bob through the quantum channel. However, the $M-1$ dishonest participants cannot distinguish the decoy qubits from the entangled qubits in  $P_{1}\prime$ because they are all in the maximally mixed state $\rho = (\ket{0}\bra{0} + \ket{1}\bra{1})/2$. Therefore, similar to the scenario with one dishonest participant, the $M-1$ malicious participants essentially act as outside eavesdroppers. As analyzed previously, their malicious behavior will inevitably be detected by Alice during the eavesdropping check. Thus, this attacking strategy is ineffective in this protocol.

\section{Efficiency analysis and comparison} \label{eff}

In this section, we perform an efficiency analysis of the proposed protocol and compare its performance with typical multi-party SQSS protocols.

To compute the qubit efficiency of the proposed scheme, we use the following formula:

\begin{equation}
    \eta = \frac{c}{q}.
\end{equation}
Here, $c$ represents the length of the classical shared secret, and $q$ represents the total number of qubits generated in the protocol, including those prepared by the classical participants.

In our scheme, Alice wants to share a classical bit string $S$ of length $N$ with $M$ participants. Therefore, $c = N$. For that, Alice generates $2N$ entangled states of $M$ particles. She also needs to generate $M$ groups of decoy Bell pairs, with each group containing $K$ Bell pairs, resulting in $2K$ qubits per group. As our protocol is a one-way communication scheme, the classical participants do not generate any qubits to send back to Alice. Therefore, the total number of qubits generated during the protocol is:

\begin{equation}
    q = (2N + 2K)M,
\end{equation}. 
if we set $K=N$, the qubit efficiency of the proposed scheme is:

\begin{equation}
    \eta = \frac{N}{4NM} = \frac{1}{4M}.
\end{equation}
The qubit efficiency of Refs. \cite{Li2010, Li2013, CHUNWEI2013, Xie2015, Yu2017, Li2020, Younes2024} can be calculated in a similar manner, and the specific results are shown in Table \ref{table 1}.

We now compare our protocol with similar multi-party schemes. The details of the comparison are shown in Table \ref{table 1}. From this comparison, it is evident that our protocol offers several advantages over the other schemes:

\begin{enumerate}
    \item As the protocol is a one-way communication scheme, the classical participants do not need to equip themselves with expensive quantum devices such as  photon number splitter (PNS) and wavelength filter (WF) to mitigate the Trojan horse attack, aligning with the original intent of a semi-quantum environment; unlike in the other protocols. 
    \item The classical participants require only two capabilities instead of three, and they do not need preparation devices to generate qubits, unlike in other protocols.
    \item Our protocol exhibits a higher qubit efficiency than most multi-party SQSS scheme.
    \item The protocol is more practical as it allows Alice to choose the content of the secret she want to share with the $M$ participants, rather than sharing a random secret like some other protocols. 
\end{enumerate}

\begin{table}
\caption{\label{table 1} Comparison of the protocol with other Multiparty SQSS protocols.}
\resizebox{\textwidth}{!}{%
\begin{tabular}{cccccccc}
\toprule
 Protocols & Quantum resources & \thead{Classical participants' \\ ability} & \thead{Mitigate the \\ Trojan horse attack} & Mitigate the DCNA & Sharing secret & Qubit efficiency \\
\midrule
 Li et al.'s \cite{Li2010} & Multi-qubits entangled states & \makecell{(1) Generate $\ket{0}$ or $\ket{1}$ \\ (2) Measure in $Z$ \\ (3) Reflect} & No & Yes & Unspecific & $\cfrac{1}{2^{M}(3M+2)}$ \\
 \hline
 Li et al.'s \cite{Li2013} & Single qubits & \makecell{(1) Generate $\ket{0}$ or $\ket{1}$ \\ (2) Measure in $Z$ \\ (3) Reflect} & No & Yes & Unspecific & $\cfrac{1}{2^{M}(3M)}$ \\
  \hline
 Yang et al.'s \cite{CHUNWEI2013} & Single qubits & \makecell{(1) Generate $\ket{0}$ or $\ket{1}$ \\ (2) Measure in $Z$ \\ (3) Reflect} & No & Yes & Unspecific & $\cfrac{1}{6M}$ \\
  \hline
 Xie's et al.'s \cite{Xie2015} & Multi-qubits entangled states & \makecell{(1) Generate $\ket{0}$ or $\ket{1}$ \\ (2) Measure in $Z$ \\ (3) Reflect} & No & Yes & Specific & $\cfrac{1}{2^{M-1}(3M+2)}$ \\
 \hline 
 Yu et al.'s \cite{Yu2017} & Multi-qubits entangled states & \makecell{(1) Generate $\ket{0}$ or $\ket{1}$ \\ (2) Measure in $Z$ \\ (3) Reflect} & No & Yes & Unspecific & $\cfrac{1}{6M+4}$ \\
 \hline
 Li et al.'s \cite{Li2020} & Bell states & \makecell{(1) Generate $\ket{0}$ or $\ket{1}$ \\ (2) Measure in $Z$ \\ (3) Reflect} & No & Yes & Unspecific & $\cfrac{1}{5M}$ \\
 \hline
 Ye et al.'s \cite{Li2024} & \makecell{(1) Multi-qubits entangled states \\ (2) Single qubits } & \makecell{(1) Generate $\ket{0}$ or $\ket{1}$ \\ (2) Measure in $Z$ \\ (3) Reflect} & No & Yes & Specific & $\cfrac{1}{3M+1}$ \\
 \hline
 Younes et al.'s \cite{Younes2024} & \makecell{(1) Multi-qubits entangled states \\ (2) Single-qubits} & \makecell{(1) Generate $\ket{0}$ or $\ket{1}$ \\ (2) Measure in $Z$ \\ (3) Reflect} & Yes & No & Specific & $\cfrac{1}{3M}$ \\
 \hline
 Our protocol & \makecell{(1) Multi-qubits entangled states \\ (2) Bell states } & \makecell{(1) Measure in $Z$ \\ (2) Perform $H$} & Yes & Yes & Specific & $\cfrac{1}{4M}$ \\

\bottomrule
\end{tabular}%
}

\end{table}

\section{Conclusion} \label{conclusion}

In this paper, we introduced a novel multi-party semi-quantum secret sharing (SQSS) scheme based on $M$-particle entangled states. This protocol improves upon our previous scheme by enabling classical participants to measure in the $Z$ basis and perform the Hadamard operation. Instead of using decoy photons, we leverage the quantum property between Bell states and the Hadamard operation to detect eavesdroppers. These changes address two main issues: (1) the security weakness of Younes et al.'s protocol against the double CNOT attack (DCNA), and (2) the need for a fully one-way communication scheme via the quantum channel to ensure robustness against Trojan horse attacks. Additionally, we demonstrated that our protocol withstands other types of attacks, including intercept-resend and collective attacks, from both internal and external sources. Our protocol offers advantages over other multi-party SQSS schemes, including better qubit efficiency and allowing Alice to control the shared secret's content. Furthermore, with recent advancements in preparing high-fidelity entangled and Bell states using IBM quantum computers \cite{Cruz2019, Sisodia2020}, our protocol is feasible for implementation with current technology.

\bibliographystyle{plainnat}
 \bibliography{msqss_ref}

\end{document}